\documentclass[11p,reqno]{amsart}

\usepackage[T1]{fontenc}
\usepackage{graphicx}
\usepackage{amssymb,amsthm,amsmath,mathrsfs,bm,braket,marginnote}
\usepackage{enumerate}
\usepackage{appendix}
\usepackage[colorlinks=true, pdfstartview=FitV, linkcolor=blue, citecolor=blue, urlcolor=blue]{hyperref}
\usepackage{multirow}

\usepackage{pgf}
\usepackage{pgfplots}
\usepackage{tikz}
\usetikzlibrary{arrows,calc}
\usepackage{verbatim}
\usetikzlibrary{decorations.pathreplacing,decorations.pathmorphing}
\usepackage[numbers,sort&compress]{natbib}
\usepackage{dsfont}
\numberwithin{equation}{section}
\newtheorem{theorem}{Theorem}[section]

\newtheorem{proposition}[theorem]{Proposition}

\theoremstyle{remark}
\newtheorem{remark}[theorem]{Remark}

 \reversemarginpar

\begin{document}

\title[]{Circulant $L$-ensembles in the thermodynamic limit}

\author{Peter J. Forrester}
\address{School of Mathematical and Statistics, ARC Centre of Excellence for Mathematical and Statistical Frontiers, The University of Melbourne, Victoria 3010, Australia}
\email{pjforr@unimelb.edu.au}

%\subjclass[2010]{15B52, 15A15, 33E20}
\date{}

\dedicatory{}

\keywords{}

\begin{abstract}
$L$-ensembles are a class of determinantal point processes which can be
viewed as a statistical mechanical systems in the grand canonical ensemble.
Circulant $L$-ensembles are the subclass which are locally translationally
invariant and furthermore subject to periodic boundary conditions. Existing
theory can very simply be specialised to this setting, allowing for the
derivation of  formulas  
for the system pressure, and the correlation kernel, in the thermodynamic limit.
For a one-dimensional domain, this is possible when the circulant matrix is both
real symmetric, or complex Hermitian. The special case of the former having a Gaussian functional
form for the entries is shown to correspond to free fermions at finite temperature, and be
generalisable to higher dimensions. A special case of the latter is shown to be
the statistical mechanical model introduced by Gaudin to interpolate between Poisson and unitary
symmetry statistics in random matrix theory. It is shown in all cases that the compressibility
sum rule for the two-point correlation is obeyed,
and  the small and large distance asymptotics of the latter are considered. Also, a conjecture relating the
asymptotic form of the hole  probability to the pressure is verified.
\end{abstract}

\maketitle

\section{Introduction}

Point processes of the type familiar in statistical mechanics consist
of $N$ indistinguishable particles confined to a domain $\Omega$.
Let particle $j$ have coordinate $x_j$, and denote a configuration of the
$N$-particles by $\boldsymbol  X_N$. Each configuration of particles
in the domain is specified by a probability density function
$p^{(N, \Omega)}(\boldsymbol X_N)$.
Important in both applications and for illustrative purposes is the case that $\Omega$
is an interval $[0,L]$ divided into $M$ lattice points at $\tau j / M$,
$j=1,\dots,M$, and where $\tau:=L/M$. Then the probability density function takes on
a discrete set of values, and with the particle coordinates ordered
\begin{equation}\label{0.0a}
0 < x_1 \le x_2 \le \cdots \le x_N \le M,
\end{equation}
is a probability,
\begin{equation}\label{0.0}
p^{(N, \Omega)}(\boldsymbol  X_N) = {\rm Pr} \, (\boldsymbol  X_N).
\end{equation}
Parameterising $\Omega$ by $\boldsymbol  M = \{1,\dots,M \}$, and each $x_j$ by an
integer $n_j$ such that $x_j = \tau n_j/ M$, $\boldsymbol  X_N$ can be
regarded as a subset of $\boldsymbol  M$ of size $N$. 

Fundamental to the statistical description of point processes are the $k$-point correlation
functions $\{ \rho_{(k)} \}$. In the continuous setting, with no ordering of the particle coordinates
assumed, these are specified in terms of the probability density function by
\begin{equation}\label{0.2}
\rho_{(k)}^{(N, \Omega)}(\boldsymbol  X_k) = N (N - 1) \cdots (N - k + 1)
\int_\Omega dx_{k+1} \cdots \int_\Omega dx_{N} \, p^{(N, \Omega)}(\boldsymbol  X_N).
\end{equation}
The case $k=1$ gives the particle density, with the characterising property that
$\int_a^b \rho_{(1)}(x) \, dx$ is equal to the expected number of particles in the interval
$[a,b]$. More generally the ratio
\begin{equation}\label{0.3}
\rho_{(k)}^{(N, \Omega)}(\boldsymbol  X_k) / \rho_{(k-1)}^{(N, \Omega)}(\boldsymbol  X_{k-1})
\end{equation}
has the interpretation of corresponding to the particle density at $x_k$, given there are particles
at $x_1,\dots, x_{k-1}$.
In the discrete setting, with an ordering convention such as (\ref{0.0a}) and $\boldsymbol  X_k$
regarded as a subset,
an appropriate modification of (\ref{0.2}) is to define
\begin{equation}\label{0.4}
{\rho}_{(k)}^{(N, \boldsymbol  M)}(\boldsymbol  X_k) =  {\rm Pr} \, (\boldsymbol  X_k) =
\sum_{ \boldsymbol  Y_{N-k}: |\boldsymbol  X_k| + |\boldsymbol  Y_{N-k}| = N}
p^{(N, \boldsymbol  M)}(\boldsymbol  X_k \cup \boldsymbol Y_{N-k}).
\end{equation}
Note the absence of the combinatorial factor $ N (N - 1) \cdots (N - k + 1)$ which is present in
(\ref{0.2}). This is in keeping with subsets not recording order.

The theme of the present paper relates to the circumstance that the probability density function
$p^{(N, \Omega)}$ has the particular functional form
\begin{equation}\label{0.1a}
p^{(N, \Omega)}(\boldsymbol  X_N) = \det [ K^{(N, \Omega)}(x_j, x_l) ]_{j,l=1}^N
\end{equation}
for some function $ K^{(N, \Omega)}(x, y)$ referred to as the correlation kernel.
Note that for repeated elements in $\boldsymbol  X_N$ (or equivalently, particles at the
same site),  $p^{(N, \Omega)} = 0$.
Moreover, we want this 
structure to be maintained upon forming the $k$-point correlation (\ref{0.2}), so that
\begin{equation}\label{0.1b}
\rho_{(k)}^{(N, \Omega)}(\boldsymbol  X_k)  = \det [ K^{(N, \Omega)}(x_j, x_l) ]_{j,l=1}^k,
\end{equation}
where $ K^{(N, \Omega)}(x,y)$ is the same function as in (\ref{0.1a}).
Such point processes are termed determinantal. Although this class may seem very restrictive, 
it has been known for some time to include a number of prominent model systems in
mathematical physics. Examples include free fermion many body wave functions in quantum mechanics \cite{Hu40,Ma75}), 
unitary invariant random matrix ensembles \cite{Dy62a,Gi65}), statistical mechanics of log-potential
Coulomb systems at a special coupling \cite{AJ81,Ga66,Ga85},
statistical mechanics of certain two-dimensional lattice
models \cite{Ka67}, Dyson Brownian motion on certain matrix spaces \cite{Dy62b},
and non-intersecting random walkers \cite{Fi84}. As a non-exhaustive list of reviews and
extended accounts of these examples and more, we reference
\cite{LM66,Me91,Fo98a,BO00,Jo02,Bo06,HKPV08,TSZ08,Fo10,Ka16,DDMS19}.

Reviews and extended accounts of theoretical developments of determinantal
point processes, often incorporating as well many examples from mathematical
physics and other settings too, are similarly numerous. Again as a 
non-exhaustive list we reference
\cite{So00,BO00,DV03, Ly03,ST03,BDF10,Bo11,KT12,LMR15,PS11,BQ17,KS19}.
Here the works \cite{KT12,LMR15} place an emphasis on properties of determinantal
point processes of particular relevance to machine learning and statistical inference.
These are thus outside of the earlier applications in mathematical physics. 
Distinguishing the applications in mathematical physics is what in statistical mechanics
is referred to as the thermodynamic limit -- this is when the number of particles and
system volume go to infinity simultaneously such that the average density is a constant.
Among determinantal point processes is a special structure when the correlation kernel
has a difference property $ K^{(N, \Omega)}(x_j, x_l) =  K^{(N, \Omega)}(x_j -x_l,0)$.
Suppose in addition that the correlation kernel is periodic of period $L$, where $L$ is the length of the
interval specifying $\Omega$. Then it turns out that a formulation in what in statistical
mechanics terminology is called the grand canonical ensemble --- termed
$L$-ensembles in the theory of determinantal point processes --- allows for an analytic
treatment in which the thermodynamic limit is readily computed. 
This working leads to the functional  form for the limiting correlation kernel
\begin{equation}\label{17.4c'} 
  K^{(\infty)}( X,Y)   =
  {z } \int_{-\infty}^\infty {e^{2 \pi i (Y-X) s} \lambda^{(\infty)}(s)  \over 1 +  z  \lambda^{(\infty)}(s)} \, ds
 \end{equation}
 obtained in (\ref{17.4c}) below. Here $z$ is the fugacity and $\lambda^{(\infty)}(s) \ge 0$ is
 determined by the functional form of the entries of $L^{(M)}$.
The mechanism underlying this calculation 
 is that the matrix determining the correlation kernel is circulant.  A development
of the consequences of a circulant  structure from a viewpoint in statistical mechanics
 is the explicit theme of the
present work.

From a technical perspective, this undertaking is straightforward: existing formulas \cite{BO00} suffice for the
general cases, and moreover special cases of the necessary working can already be found
in the literature \cite{Fo93a}. On the other hand, no one work logically develops circulant
$L$-ensembles, nor does any single work isolate physically motivated examples from this
viewpoint. Filling this gap in the literature is the contribution of the present work. In section 2 we begin by revising the formalism of $L$-ensembles, and in particular the formula for the correlation kernel in terms of a particular matrix $L^{(M)}$. When this matrix is circulant, the correlation kernel can be computed explicitly, and evaluated in certain limits.
First considered is a continuum limit when the number of lattice sites goes to infinity but the interval stays fixed --- this is equivalent to a grand canonical formalism defined
to begin on the interval; subsection 2.5 gives a direct approach in this setting. However this does not correspond to a thermodynamic limit as the expected total number of
particles is still finite. True thermodynamic limits are considered next, starting with 
a limit when the lattice spacing stays fixed with the number of lattice sites going to infinity, which then is an example of a lattice gas.
Taking the lattice spacing to zero then specifies a statistical state defined on the whole real line. These are considered in subsection 2.4 in the case of $L^{(M)}$ being real symmetric, and in subsection 2.5 when this matrix is complex Hermitian. 

Section 3 begins by showing that $L$-ensembles in the thermodynamic limit of the type considered in section 2 all obey the compressibility sum rule from the theory of fluids. This has significance in relation to the gap probability --- that is the probability that a prescribed interval is free of eigenvalues, as in this setting it has previously been conjectured that the leading asymptotic form of this probability is equal to the exponential of minus of the pressure. Known asymptotics of Toeplitz operators allow this to be checked for the circulant $L$-ensembles in the thermodynamic limit. In subsection 3.4 the particular example of a real circulant $L$-ensemble with the underlying matrix having entries given by a Gaussian functional form is considered. The function $\lambda^{(\infty)}(s)$ in (\ref{17.4c}) is then also a Gaussian. Upon appropriate identification of parameters, (\ref{17.4c'})
can then be identified with that for free fermions in one dimension at finite temperature,
\begin{equation}\label{K.14}
  		K^{(\infty)}(X,Y) = {1 \over 2 \pi} \int_{-\infty}^\infty {e^{i (Y - X)k} \over
  			e^{\beta (k^2 - \mu)} + 1} \, dk.
  			\end{equation}
Here $\beta$ is the inverse temperature, and $\mu$ is the chemical potential.
In the following subsection, a specific example of the complex Hermitian circulant $L$-ensemble is given which can be identified with a statistical mechanical model introduced by Gaudin for application in random matrix theory. The final subsection considers a higher dimensional analogue of the Gaussian functional form, and the resulting correlation kernel is identified with that for free fermions in $d$-dimensions at finite temperature.

\section{$L$-ensembles}
\subsection{Formalism}\label{F}
Consider the discrete setting specified in the opening paragraph of the
Introduction, using a subset viewpoint of $\mathbf X_N$,
but modified so that the value of $N$ can vary. This is done
by specifying that for each $N=0,\dots,M$
\begin{equation}\label{1.0}
{\rm Pr} \, (\boldsymbol  X_N) = {z^N  \over \Xi^{(\boldsymbol M)}(z)}
\det \Big [ L^{(\boldsymbol M)}(x_j,x_l) \Big ]_{j,l=1}^N,
\end{equation}
where $z > 0$ is a parameter, and $ \Xi^{(\boldsymbol M)}(z)$ the normalisation.
Such probabilities on subsets of $\boldsymbol M$ are referred to as $L$-ensembles.
In statistical mechanics, $ \Xi^{(\boldsymbol M)}(z)$ is referred to as the grand canonical partition function
and $z$ as the fugacity.
From expansion properties of the determinant it is easy to see that with $ L^{(\boldsymbol M)} :=  [L^{(\boldsymbol M)}(j,l)]_{j,l=1}^M$
\begin{equation}\label{1.1}
\Xi^{(\boldsymbol M)}(z) := \sum_{N=0}^{M} z^N  \sum_{\boldsymbol X_N \subset \boldsymbol M}
\det \Big [ L^{(\boldsymbol M)}(x_j,x_l) \Big ]_{j,l=1}^N =
\det \Big ( \mathbb I_M + z L^{(\boldsymbol M)} \Big).
\end{equation}
In this setting the $k$-point correlation function is defined by
\begin{equation}\label{1.2}
\rho_{(k)}^{(\boldsymbol M)}(\boldsymbol  X_k)  = 
\sum_{s=0}^{M-k} \sum_{\boldsymbol Y_s \subset \boldsymbol M \atop
|\boldsymbol Y_s | = s} {\rm Pr} \, \Big (\boldsymbol  X_k \cup \boldsymbol  Y_s \Big ).
\end{equation}

All $L$-ensembles are in fact determinantal point processes.

\begin{proposition} (Macchi \cite{Ma75})
For an $M \times M$ matrix $J$, let $(J)_{\boldsymbol  X_k}$ denote the
$k \times k$ submatrix formed from the entries in the rows and columns
labelled by $\boldsymbol  X_k$.  With this notation, we have
\begin{equation}\label{1.3}
\rho_{(k)}^{(\boldsymbol M)}(\boldsymbol  X_k)  =  \det (K^{(\boldsymbol M)})_{\boldsymbol  X_k},
\end{equation}
where 
\begin{equation}\label{1.4}
K^{(\boldsymbol M)} =  z L^{(\boldsymbol M)} (\mathbb I_M + z L^{(\boldsymbol M)})^{-1}.
\end{equation}
\end{proposition}

\begin{proof}
The sum in (\ref{1.2}) permits an evaluation analogous to (\ref{1.1}), implying
\begin{equation}\label{1.5}
\rho^{(\boldsymbol M)}_{(k)}(\boldsymbol  X_k)  =  {1 \over  \Xi^{(\boldsymbol M)}(z)} \det \Big ( \bar{\mathbb I}_M^{\boldsymbol  X_k} + z  L^{(\boldsymbol M)} \Big ),
\end{equation}
where $ \bar{\mathbb I}_M^{\boldsymbol  X_k} $ denotes the $M \times M$ identity matrix with diagonal entries $\boldsymbol  X_k$ each replaced by zero.
Writing $ \bar{\mathbb I}_M^{\boldsymbol  X_k} =  \mathbb I_M -   \mathbb I_M^{\boldsymbol  X_k}$, substituting the determinant formula
(\ref{1.1}) for $  \Xi^{(\boldsymbol M)}(z)$ and
using the multiplicative property of the determinant shows
\begin{equation}\label{1.5a}
\rho^{(\boldsymbol M)}_{(k)}(\boldsymbol  X_k)   = \det \Big (  \mathbb I_M -  {\mathbb I}_M^{\boldsymbol  X_k}  (\mathbb I_M + z L^{(\boldsymbol M)}  )^{-1} \Big ).
\end{equation}
Introducing  $K^{(M)}$ as defined by (\ref{1.4}), simple manipulation of (\ref{1.5a}) then shows
\begin{equation}\label{1.5b}
\rho^{(\boldsymbol M)}_{(k)}(\boldsymbol  X_k)  
 =  \det \Big (  \bar{\mathbb I}_M^{\boldsymbol  X_k}  +  {\mathbb I}_M^{\boldsymbol  X_k}  K^{(\boldsymbol M)} \Big ).
\end{equation}
 The result (\ref{1.3}), (\ref{1.4}) now follows by expansion
properties of the determinant.
\end{proof}

\subsection{Toeplitz $L$-ensembles}
For (\ref{1.0}) to be well defined, we must have 
\begin{equation}\label{15.0}
\det \Big [ L^{(\boldsymbol M)}(x_j,x_l) \Big ]_{j,l=1}^N \ge 0.
\end{equation}
Requiring too that $ L^{(\boldsymbol M)}$ as specified in (\ref{1.1}) be Hermitian, we know from linear algebra that
this is equivalent to $ L^{(\boldsymbol M)}$ being semi-positive definite, and thus for any $\mathbf c = (c_1,\dots, c_M)$,
that
\begin{equation}\label{15.1}
\mathbf c   L^{(\boldsymbol M)} \mathbf c^\dagger \ge 0.
\end{equation}

Let us investigate (\ref{15.1}) in the setting that $L^{(\boldsymbol M)} (j, l)$ has the difference property $L^{(\boldsymbol M)} (j, l) = L^{(\boldsymbol M)} (j-l,0)$ so that the
configurations $\boldsymbol  X_N$ (thinking now in the particle picture) specified by (\ref{1.0}) are all translationally invariant.
The difference property is equivalent to saying that $L^{(\boldsymbol M)}$ is a Toeplitz matrix. Introduce $f(\zeta)$ as a generating function for the independent entries of $L^{(\boldsymbol M)}$, so that
\begin{equation}\label{15.2}
f(\zeta) = \sum_{j= 0}^{M-1}  L^{(\boldsymbol M)}(j,0) \zeta^j, \qquad    L^{(\boldsymbol M)}(p,0) = \int_0^1 f(e^{2 \pi i x}) e^{-2 \pi  i p x} \, dx.
\end{equation}
Using the latter formula herein in (\ref{15.1}) shows
\begin{equation}\label{15.3}
\mathbf c   L^{(\boldsymbol M)} \mathbf c^\dagger  = \int_0^1 f(e^{2 \pi i x})  | C(e^{2 \pi i x}) |^2 \, dx, \qquad C(e^{2 \pi i x}) = \sum_{j=0}^{M-1} c_j  e^{2 \pi i j x}.
\end{equation}
It follows that a sufficient condition for $L^{(\boldsymbol M)}$ to satisfy (\ref{15.1}) is that $f(e^{2 \pi i x})$ be real  --- which is guaranteed by
the assumption that  $L^{(\boldsymbol M)}$ be Hermitian, and thus $L^{(\boldsymbol M)} (l, j) = \overline{L^{(\boldsymbol M)} (j, l)}$ --- and 
furthermore that $f(e^{2 \pi i x})$ be non-negative.

\begin{remark}
1.~It is not a necessary condition that $ L^{(\boldsymbol M)}$ be Hermitian for (\ref{15.0}) to hold true.
See \cite{Br05} for examples of  tridiagonal Toeplitz matrices of this type. \\
2.~For $\kappa := (\kappa_1,\dots, \kappa_n)$ a partition of non-negative integers, the Schur polynomial
is defined by
\begin{equation}\label{15.4}
s_\kappa(z_1,\dots,z_n) = { \det [ z_j^{\kappa_l + n - l} ]_{j,l=1}^n  \over  \det [ z_j^{ n - l} ]_{j,l=1}^n }.
\end{equation}
This definition extends to general tuples $\kappa$. Standard manipulations used in random matrix theory
(see e.g.~\cite[Exercises 5.4 q.1]{Fo10}), starting with the formula for $L^{(\boldsymbol M)}(p,0)$ in (\ref{15.2})
substituted in the LHS of (\ref{15.0}) show
\begin{multline}\label{15.5}
{1 \over N!} \det \Big [ L^{(\boldsymbol M)}(x_j,x_l) \Big ]_{j,l=1}^N = \int_0^1 dt_1 \, f(e^{2 \pi i t_1}) \cdots \int_0^1 dt_N \,
 f(e^{2 \pi i t_N}) 
|  s_\kappa(e^{2 \pi i t_1},\dots, e^{2 \pi i t_N}) |^2  \\ \times \prod_{1 \le j < l \le N} | e^{2 \pi i t_l} - e^{2 \pi i t_j} |^2,
\end{multline}
where $\kappa = (x_1,\dots,x_N)$.
Given that $ f(e^{2 \pi i t})$ is non-negative, this gives another way of seeing that $ \det \Big [ L^{(\boldsymbol M)}(x_j,x_l) \Big ]_{j,l=1}^N$
is non-negative. \\
3.~The Toeplitz matrix $[L^{(\boldsymbol  M)}(j-l,0)]_{j,l=1}^M$ is {\it not} a random matrix. Rather in the $L$-ensemble formalism beginning with
(\ref{1.0}) the matrix elements $L^{(M)}(x,y)$ are a prescribed functional form. For works which do address random
Toeplitz matrices, as a non-exhaustive list we draw attention to \cite{DGK09,MSS17,Bo18,MS19,Bo20}.

\end{remark}

\subsection{Real circulant $L$-ensembles}
Suppose that in addition to the difference property, the entries of $L^{(\boldsymbol M)}$ are periodic of period $M$,
which is characteristic of periodic boundary conditions, or equivalently the original interval $[0,L]$
being identified with the circumference of a circle.
A systematic way for the entries of $L^{(\boldsymbol M)}$ to have this periodicity, and to have a well
defined scaling limit for large $M, L$, is to measure
distance  as chord length, so that
\begin{equation}\label{15.6}
 L^{(\boldsymbol M)}(x_j,x_l) =  L^{(\boldsymbol M)}(x_j - x_l,0) = g\Big ((L/\pi) \sin(\pi (x_j - x_l)/M) \Big ), 
 \end{equation}
 for some $g(u)$ even and real valued. Note that this latter condition implies 
 that the matrix $[L^{(\boldsymbol  M)}(j-l,0)]_{j,l=1}^M$ is real symmetric. With both the
 difference and periodicity property of the elements, the matrix $ L^{(\boldsymbol M)}$ is referred to as being circulant.
 The significance of this extra structure is that the normalised eigenvectors of all circulant matrices are
 independent of $g$ in (\ref{15.6}) and given by
 \begin{equation}\label{15.7}
 {1 \over \sqrt{M}}  e^{-2 \pi i p_0/M} (1, e^{2 \pi i p/M},  e^{4 \pi i p/M},\dots, e^{2 \pi i (M - 1)/M})^T, 
 \quad p=0,\dots,M-1,
  \end{equation}
  for any integer $p_0$, which we take as equal to 
  $ \lfloor - M/2 \rfloor + 1$.
  As a consequence, the normalisation $\Xi^{(M)}(z)$ and matrix elements $K^{(M)}(x,y)$ can be computed explicitly.
  
\begin{proposition}\label{p2.3}
In the setting of (\ref{15.6}) we have
 \begin{equation}\label{15.8}  
 \Xi^{(\boldsymbol M)}(z) = \prod_{p= \lfloor - M/2 \rfloor +1}^{\lfloor M/2 \rfloor}(1 + z \lambda_p), \quad \lambda_p = \sum_{s= \lfloor - M/2 \rfloor +1}^{ \lfloor M/2 \rfloor} g(L
 \sin(\pi s/M) ) e^{2 \pi i p s /M} =
 \tilde{f}(e^{2 \pi i p/M}),
  \end{equation}
 where $ \tilde{f}(z)$ is the particular generating function for the entries (\ref{15.6})
  \begin{equation}\label{15.6a}
  \tilde{f}(z) :=  \sum_{s= \lfloor - M/2 \rfloor +1}^{ \lfloor M/2 \rfloor} g(L
  \sin(\pi s/M) ) z^s
  \end{equation}
  (cf.~(\ref{15.2})), and
 \begin{equation}\label{15.9}    
 K^{(\boldsymbol M)}(x,y) =  {z \over M} \sum_{p=\lfloor - M/2 \rfloor + 1}^{\lfloor M/2 \rfloor} { e^{2 \pi i (y - x) p /M} \lambda_p \over 1 + z \lambda_p}.
 \end{equation}
 \end{proposition}
 
 \begin{proof}
 The form of the matrix entries (\ref{15.6}) and the fact that the eigenvectors are given by (\ref{15.7})
 implies the formula for the eigenvalues $\lambda_p$ in (\ref{15.8}), while the determinant formula
 in (\ref{1.1}) implies the first equality therein.
 
 Using the diagonalisation formula
  \begin{equation}\label{15.10}   
   L^{(\boldsymbol M)} =   V^{(\boldsymbol M)} {\rm diag} \, (\lambda_1,\dots, \lambda_N)  (V^{(\boldsymbol M)} )^\dagger,
 \end{equation}   
 where the columns of $V^{(\boldsymbol M)} $ are given by the normalised eigenvectors $\{ \mathbf v_p \}$ of $ L^{(\boldsymbol M)} $,
 we have from (\ref{1.4}) that whenever $L^{(\boldsymbol M)} $ is Hermitian
  \begin{equation}\label{15.11}    
 K^{(\boldsymbol M)}(x,y) = z \sum_{p=0}^{M-1} { (\mathbf v_p)^{(x)}   (\bar{\mathbf v}_p)^{(y)}  \lambda_p \over 1 + z \lambda_p}.
 \end{equation}
 Here $(\mathbf v_p)^{(x)} $ denotes the $x$-th component of $\mathbf v_p$ and $ K^{(\boldsymbol M)}(x,y)$ denotes the
 entry in row $x$, column $y$ of the matrix $ K^{(\boldsymbol M)}$. The fact that for $L^{(\boldsymbol M)} $ a circulant matrix,
 the eigenvectors are given by (\ref{15.7}), implies (\ref{15.9}).
 \end{proof}
 
 \subsection{Large $M$ limits}
 In the first paragraph of the Introduction, the parameter $M$ was introduced as the number of
 equally spaced lattice point used to discretise a line of length $L$. Two large $M$ limits
 are therefore suggested. One is when $M$ and $L$ go simultaneously to infinity and the
 lattice spacing remains fixed. The other is when $M \to \infty$ but $L$ remains fixed, which reclaims
 the point process on the continuum interval $[0,L]$. We will consider each separately.
 
 \subsubsection{$M \to \infty$, $L/M$ fixed}
\begin{proposition}\label{p2.4} 
 Let $\tau := L/M$ denote the lattice spacing,
 and with $g(u)$ as in (\ref{15.6}) let
  \begin{equation}\label{ft1} 
  \tilde{f}^{(\infty)}(\zeta) :=
  \sum_{s=-\infty}^\infty g(\tau s) \zeta^s,
 \end{equation}
   where the sum is assumed to be well defined for $|\zeta|=1$.
 We have
   \begin{equation}\label{16.1}  
 \tau \beta P^{(\tau)}  :=  \lim_{M \to \infty \atop  \tau \: {\rm fixed}} {1 \over M} \log  \Xi^{(\boldsymbol M)}(z) = \int_{-1/2}^{1/2} \log \Big ( 1 + z  \tilde{f}^{(\infty)} (e^{2 \pi i t}) \Big ) \, dt,
 \end{equation}
 where the notation $\tau \beta P^{(\tau)}$ relates to the interpretation of the limit in terms of the pressure from statistical mechanics, and
  \begin{equation}\label{16.2}     
  K^{(\tau)}(x,y)  :=  \lim_{M \to \infty \atop  \tau \: {\rm fixed}}    K^{(\boldsymbol M)}(x,y) = z \int_{-1/2}^{1/2} {e^{2 \pi i (y - x) t}  \tilde{f}^{(\infty)}(e^{2 \pi i t}) \over 1 +
   z  \tilde{f}^{(\infty)}(e^{2 \pi i t}) } \, dt.
 \end{equation}   
\end{proposition}
 
 \begin{proof}
 We see from the  definition (\ref{15.6a}) that
  \begin{equation}\label{16.2a}
 \lim_{M \to \infty \atop  \tau \: {\rm fixed}}  
 \tilde{f}(\zeta) =  \tilde{f}^{(\infty)}(\zeta),
 \end{equation}
as specified by (\ref{ft1}).
The stated formulas now   
follow by recognising the appropriate sums in
 Proposition \ref{p2.3}  as Riemann approximations to
 definite integrals. 
 \end{proof}

 As a check, we see from (\ref{1.1}), (\ref{15.8}) and (\ref{15.9}) that
  \begin{equation}\label{16.3}  
  z {\partial \over \partial z }  \lim_{M \to \infty} {1 \over M} \log  \Xi^{(\boldsymbol M)}(z)  = \lim_{M \to \infty} {1 \over M} {\rm Tr} \,   K^{(\boldsymbol M)},
  \end{equation}
 which requires too the validity of interchanging the derivative with the limit on the LHS. Since the RHS is equal to 
 $   \lim_{M \to \infty}    K^{(\boldsymbol M)}(x,x)$ independent of $x$ in the present setting of periodic boundary conditions, the results
 (\ref{16.1}) and (\ref{16.2}) are consistent with (\ref{16.3}). Note too that this tells us that the expected number of particles per
 lattice site in the $M \to \infty$ limit, $\rho^{(\tau)} $ say, is given in terms of the fugacity $z$ by
  \begin{equation}\label{16.4}  
  \rho^{(\tau)} =  z \int_{-1/2}^{1/2} {\tilde{f}^{(\infty)}(e^{2 \pi i t}) \over 1 +
   z  \tilde{f}^{(\infty)}(e^{2 \pi i t}) } \, dt.
  \end{equation} 
  In particular, since $ \tilde{f}^{(\infty)}(z) \ge 0$ for $|z| = 1$,
  the requirement that $0 \le \rho^{(\tau)} \le 1$ is evident.
  
   \subsubsection{$M \to \infty$, $L$ fixed, followed by $L \to \infty$}
   Starting with $L$ fixed we must have that the lattice labels $x,y$ are also suitably
   scaled to correspond to points in the interval, $Lx/M \to X, Ly/M \to Y$,
   and furthermore $z$ must be scaled $z \mapsto L z/M$ to allow for the
   sums in (\ref{0.4}) to correspond to Riemann integrals. 
   The functional form (\ref{15.6}) is again appropriate, but no longer with any
   reference to $M$, so we define
    \begin{equation}\label{15.6a+} 
    L^{(L)}(X,Y) = g \Big ( (L/\pi) \sin(\pi (X -  Y)/L) \Big ).
    \end{equation}

  \begin{proposition}\label{p2.5a} 
Let $g(u)$ be as in (\ref{15.6a+}), and suppose furthermore that it be integrable on $(-L/\pi, L/\pi)$. We have
   \begin{equation}\label{17.1}  
   \log  \Xi^{(L)}(z) :=  \lim_{M \to \infty  \atop L \: {\rm fixed}} \log  \Xi^{(\boldsymbol M)}(Lz/M) = \sum_{p=-\infty}^\infty \log \Big ( 1 + z \lambda_p^{(L)} \Big ), \quad
  \end{equation}
 and
  \begin{equation}\label{17.2}     
  K^{(L)}( X,Y)  :=  \lim_{M \to \infty \atop L \: {\rm fixed}}  (M/L)  K^{(\boldsymbol M)}(M X/L,MY/L) \Big |_{z \mapsto L z/M} = 
  {z \over L} \sum_{p=-\infty}^\infty {e^{2 \pi i (Y-X) p /L} \lambda_p^{(L)}  \over 1 +  z \lambda_p^{(L)} },
 \end{equation}   
 where
 \begin{equation}\label{17.3}   
  \lambda_p^{(L)}  = \lim_{M \to \infty \atop  L \: {\rm fixed}}    {L \over M} \lambda_p = L \int_{-1/2}^{1/2} {g}\Big
  ( {L \over \pi} \sin \pi  t \Big ) e^{2 \pi i p t} \, dt.
  \end{equation}
 \end{proposition} 
   
\begin{proof}
These formulas follow from 
(\ref{15.8}) and (\ref{15.9}), by viewing (\ref{15.6a}) as a
Riemann sum and thus
 \begin{equation}\label{17.3c} 
\lim_{M \to \infty \atop  L \: {\rm fixed}}
{L \over M} 
 \tilde{f}(e^{2 \pi i p /M}) = L 
\int_{-1/2}^{1/2} {g}\Big
( {L \over \pi} \sin \pi  t \Big ) e^{2 \pi i p t} \, dt,
 \end{equation}
 which corresponds to (\ref{17.3}).

 \end{proof}

After changing variables
$t \mapsto t/L$ in (\ref{17.3}),
 the large $L$ limit of the quantities in Proposition \ref{p2.5a} are almost immediate, with the sums therein recognised
 as Riemann approximations to definite integrals.
 
   \begin{proposition}\label{p2.5b} 
   Let 
  \begin{equation}\label{17.4}   
    \lambda^{(\infty)}(s)  = \int_{-\infty}^\infty {g}(t) e^{2 \pi i s t} \, dt.
 \end{equation}      
 We have 
  \begin{equation}\label{17.4b} 
  \beta P =
\lim_{L \to \infty} {1 \over L}   \log  \Xi^{(L)}(z)  = \int_{-\infty}^\infty \log \Big ( 1 + z \lambda^{(\infty)}(s)  \Big ) \, ds,
  \end{equation}
  and 
  \begin{equation}\label{17.4c}     
  K^{(\infty)}( X,Y)   :=     \lim_{L \to \infty}  K^{(L)}( X,Y) =
  {z } \int_{-\infty}^\infty {e^{2 \pi i (Y-X) s} \lambda^{(\infty)}(s)  \over 1 +  z  \lambda^{(\infty)}(s)} \, ds.
 \end{equation} 
 \end{proposition} 

\begin{remark}
	We see from (\ref{16.1}) and (\ref{ft1}) that
	\begin{equation}
	\lim_{\tau \to 0^+} \beta P^{(\tau)}  =
	\int_{-\infty}^\infty \log (1 + z \lambda^{(\infty)}(s)) \, ds,
	\end{equation}
	thus reclaiming (\ref{17.4b}), and from (\ref{16.2}) and (\ref{ft1}) that
	\begin{equation}
	\lim_{\tau \to 0^+}  {1 \over \tau}
	K^{(\tau)}(X/\tau, Y/\tau) \Big |_{z \mapsto \tau z}=
{z } \int_{-\infty}^\infty {e^{2 \pi i (Y-X) s} \lambda^{(\infty)}(s)  \over 1 +  z  \lambda^{(\infty)}(s)} \, ds,
\end{equation} 
thus reclaiming	(\ref{17.4c}). 
Hence, with the functional form (\ref{15.6}) for the
matrix elements of $L^{(\boldsymbol M)}$, the results for the continuum can be reclaimed from the results for the lattice, upon taking the lattice spacing to zero.
\end{remark}

\subsection{Complex circulant $L$-ensembles}
In the interests of brevity, we will restrict attention in this circumstance to the setting of particles on the continuous segment $[0,L]$ in periodic boundary conditions, which can be thought of as the circumference of a circle. Measuring distance as chord length, a functional form giving rise to a complex Hermitian circulant integral operator (no longer a matrix since the domain is a continuum) is
\begin{equation}\label{C.0}
L^{(L)}(X,Y) = i h \Big ( (L/\pi) \sin(\pi(X - Y + 2 i \epsilon)/L)
\Big ), \quad \epsilon > 0,
\end{equation}
with $h(u)$ odd; cf.(\ref{15.6}). Note that this latter requirement implies $\overline{L^{(L)}(X,Y) } = L^{(L)}(Y,X)$ and thus
$L^{(L)}$ is a Hermitian matrix.
The corresponding probability density for a configuration $\boldsymbol X_N$ is
\begin{equation}\label{C.0a}
p(\boldsymbol X_N) = {z^N \over \Xi^{(L)}(z)}
\det \Big [L^{(L)}(X_j,X_k)  \Big ]_{j,k=1}^N
\end{equation}
(cf.~(\ref{1.0})). Here
\begin{equation}\label{C.0b}
\Xi^{(L)}(z) = \det ( \mathbb I + z \mathbb L),
\end{equation}
where $\mathbb L$ is the integral operator on $[0,L]$ with kernel
(\ref{C.0}); the determinant can be specified as the product over the eigenvalues.

Working directly in the continuum setting, the appropriate generalisation of (\ref{1.1}) is
\begin{equation}\label{2.39}
\Xi^{(L)}(z) = \prod_{p=-\infty}^\infty (1 + z \lambda_p),
\end{equation}
while the appropriate generalisation of (\ref{15.11}) is
\begin{equation}\label{2.40}
K^{(L)}(X,Y) = z \sum_{p=-\infty}^\infty {v_p(X) \bar{v}_p(Y) \lambda_p \over 1 + z \lambda_p}.
\end{equation}
Here $\{ \lambda_p \}$ and $\{ v_p(X) \}$ are the eigenvalues
and eigenfunctions of the integral operator $\mathbb L$ on
$[-L/2,L/2]$ (the periodicity in shifts by $L$ of (\ref{C.0}) has been
used to translate the interval) with kernel (\ref{C.0}). Thus
\begin{equation}
\mathbb L[f](X) = i \int_{-L/2}^{L/2}
h \Big ( (L/\pi) \sin(\pi(X - Y + 2i \epsilon)/L)
\Big ) f(Y) \, dY.
\end{equation}
The normalised eigenfunctions of $\mathbb L$ are
\begin{equation}\label{2.39a}
v_p(X) = {1 \over \sqrt{L}} e^{2 \pi i p X/L}, \quad
p \in \mathbb Z,
\end{equation}
which allows for (\ref{2.39}) and (\ref{2.40}) to be made explicit, and moreover for computation of the limit $L \to \infty$.

\begin{proposition}\label{p2.e}
	The formulas (\ref{2.39}) and (\ref{2.40}) hold with $v_p(x)$ given by (\ref{2.39a}) and
	\begin{equation}\label{2.40a}
	\lambda_p = i \int_{-L/2}^{L/2}
h \Big ( (L/\pi) \sin(\pi( - Y + 2i \epsilon)/L)
\Big ) e^{2 \pi i p Y/L} \, dY.
\end{equation}	
Furthermore, with
	\begin{equation}\label{e.1}
\lambda^{(\epsilon, \infty)}(s):= i \int_{-\infty}^{\infty}
h  (  - Y + 2i \epsilon  )
 e^{2 \pi i s Y} \, dY
\end{equation}
replacing $\lambda^{(\infty)}(s)$, the formulas
 (\ref{17.4b}) and (\ref{17.4c}) again hold.
\end{proposition}	

\section{Properties and examples}

\subsection{Compressibility sum rule}
In this section, in distinction to the subset viewpoint of  a configuration $\boldsymbol X_N$ used in
Section \ref{F}, it is convenient to consider $\boldsymbol X_N$ as an $N$-tuple, not to require
an ordering of the particles such as in (\ref{0.0a}), and to consider the domain as continuous.
The grand canonical ensemble formalism in this setting specifies that the probability density
function for there being $N$ particles in configuration $\boldsymbol X_N$ be given by
	\begin{equation}\label{3.1a+}
p^{(N,\Omega)}(\boldsymbol X_N) = 	{z^N \over \Xi^{(\Omega)}(z)} \det \Big [ L^{(\Omega)}(x_j, x_l) \Big ]_{j,l=1}^N
\end{equation}
where
\begin{equation}\label{3.1b+}
\qquad  \Xi^{(\Omega)}(z) = 1 +
\sum_{N =1}^\infty {z^N  } \int_{\Omega} dx_1 \cdots  \int_{\Omega} dx_N \,  \det \Big [ L^{(\Omega)}(x_j, x_l) \Big ]_{j,l=1}^N ;
\end{equation}
cf.~(\ref{1.0}) and (\ref{1.1}).
The corresponding $k$-point correlation function is given by
	\begin{multline}
	\rho_{(k)}^{(\Omega)}(\boldsymbol  X_k) =
	z^k
\Big (  k! p^{(k,\Omega)}(\boldsymbol X_k ) \\+
\sum_{n=1}^\infty z^n (n+k) \cdots (n+1)
\int_\Omega dy_1 \cdots \int_\Omega dy_n \,
 p^{(k+n,\Omega)}(\boldsymbol X_k \cup  \boldsymbol Y_n )
 \Big ).
 \end{multline}
 Suppose furthermore that the system is translationally invariant. Then
 	\begin{equation}
\int_\Omega dx_2 \,   \rho^{(\Omega)}_{(2)}(x_1,x_2)	=
{1 \over | \Omega |}
\int_\Omega dx_1 \int_\Omega dx_2 \,   \rho_{(2)}^{(\Omega)}(x_1,x_2)
= {1 \over | \Omega |}{1 \over 
 \Xi^{(\Omega)}(z) }
z {\partial^2 \over \partial z^2}
 \Xi^{(\Omega)}(z).
\end{equation}
Simple manipulation, using the fact that for a translationally invariant system
\begin{equation}
	\rho_{(1)}^{(\Omega)}(x) 
	 = {1 \over | \Omega |}
		{\partial \over \partial z} \log \Xi^{(\Omega)}(z),
\end{equation}
shows from this that
\begin{equation}
\int_\Omega dx_2 \,   
\Big ( 	\rho_{(2)}^{(\Omega)}(x_1,x_2) -
\rho_{(1)}^{(\Omega)}(x_1) \rho_{(1)}^{(\Omega)}(x_2) 
+ \delta(x_1 - x_2)  \rho_{(1)}^{(\Omega)}(x_2)
\Big )
= {1 \over | \Omega |}
\Big ( z {\partial \over \partial z} 
\Big )^2
\log \Xi^{(\Omega)}(z).
\end{equation}
Taking the limit $| \Omega | \to \infty$, assuming the limit operation can be taken inside the integral on the LHS, and the derivative operation on the RHS, then gives
\begin{equation}\label{CP}
\int_\Omega dx_2 \,   
\Big ( 	\rho_{(2)}^{(\infty)}(x_1,x_2) -
\rho^{(\infty)}_{(1)}(x_1) \rho_{(1)}^{(\infty)}(x_2) +
\delta(x_1 - x_2)  \rho_{(1)}^{(\infty)}(x_2)
\Big )
=  \Big ( z {\partial \over \partial z} \Big )^2 \beta P. 
\end{equation}
In the theory of fluids, this is referred to as the compressibility sum rule \cite{HM06}.

\begin{proposition}
	The limiting continuous determinantal point process specified by Proposition \ref{p2.5b} satisfies
	(\ref{CP}).
	\end{proposition}

\begin{proof}
	According to (\ref{1.3})
\begin{equation}\label{3.7}
\rho_{(2)}^{(\infty)}(X,Y) -
\rho_{(1)}^{(\infty)}(X) \rho_{(1)}^{(\infty)}(Y) = -
| K^{(\infty)}(X,Y) |^2.
\end{equation}
Substituting (\ref{17.4c})	shows
\begin{equation}\label{CP1}
\int_{-\infty}^\infty 
\Big ( 	\rho_{(2)}^{(\infty)}(X,Y) -
\rho_{(1)}^{(\infty)}(X) \rho_{(1)}^{(\infty)}(Y) \Big )
\,  d Y 
= - z^2 \int_{-\infty}^\infty
\Big ( {\lambda^{(\infty)}(s) \over 1 + z \lambda^{(\infty)}(s)}
\Big )^2 \, ds.
\end{equation}
Substituting (\ref{17.4b}) in the RHS of (\ref{CP}) and
subtracting $\rho_{(1)}^{(\infty)}(x)$ gives this same expression.
\end{proof}	

\begin{remark}
	1.~The same working, with $\lambda^{(\infty)}(s)$ replaced by $\lambda^{(\epsilon,\infty)}(s)$,
	verifies that the 
	continuous determinantal point process specified by Proposition \ref{p2.e} satisfies
	(\ref{CP}). \\
	2.~In the case of an infinite lattice, the integral over $x_2$ should be replaced by a sum over the lattice points, and $\beta P$ on the RHS should be replaced by $\tau \beta P$.
	Upon this modification, the results of Proposition
	\ref{p2.4} can be checked to be consistent.
	\end{remark}
		
\subsection{Gap probabilities}\label{S3.2}
For a point process defined on the real line, let $E^{(\infty)}(n;J)$ denote the probability that there are exactly $n$ particles within an interval $J$.
The case $n=0$ is referred to as the gap probability.
Introducing the generating function
\begin{equation}\label{Sh.1}
\tilde{E}^{(\infty)}(J;\xi) = \sum_{n=0}^\infty (1 - \xi)^n
{E}^{(\infty)}(n;J),
\end{equation}
it is a standard result (see e.g.~\cite[\S 9.1]{Fo10}) that
$\tilde{E}^{(\infty)}(J;\xi)$ can be written in terms of the correlation functions according to
\begin{equation}\label{Sh.1x}
\tilde{E}^{(\infty)}(J;\xi) = 1 + 
\sum_{j=1}^\infty {(-\xi)^j \over j!}
\int_J dx_1 \cdots \int_J dx_j \,
	\rho_{(j)}^{(\infty)}(x_1,\dots,x_j).
	\end{equation}

Specialise now to the case of determinantal correlations
\begin{equation}\label{Sh.2}
	\rho^{(\infty)}_{(j)}(x_1,\dots,x_j) =
	\det \Big [ K^{(\infty)}(X_{j_1}, X_{j_2}) \Big ]_{j_1,j_2=1,\dots,j}.
	\end{equation}
	Let $\mathbb K_J$
	denote the integral operator supported on $J$ with
	kernel $ K^{(\infty)}(X, Y) $.
	Then the summation
	 (\ref{Sh.2}) can be recognised as a key quantity within the Fredholm theory of integral equations
	 \cite{WW65}, namely the Fredholm determinant,
	 \begin{equation}\label{Sh.3}
	 \tilde{E}^{(\infty)}(J;\xi) =
	 \det ( \mathbb I - \xi \mathbb K_J) =
	 \prod_{j=0}^\infty (1 - \xi \lambda_j^{(J)}),
	 \end{equation}
	 where $\{ \lambda_j^{(J)}\}$ denotes the eigenvalues of $\mathbb K_J$; in fact such a quantity has
	 already appeared in (\ref{C.0b}).  Further specialise to the Toeplitz setting so that in  (\ref{Sh.2}) $K^{(\infty)}(X_{j_1}, X_{j_2}) = K^{(\infty)}(X_{j_1} - X_{j_2},0)$, and
	 introduce the Fourier transform
	  \begin{equation}\label{Sh.4}
	  \hat{K}^{(\infty)}(s) = \int_{-\infty}^\infty
	   K^{(\infty)}(x,0) e^{isx} \, dx.
	   \end{equation}
	   The asymptotic theory of Toeplitz integral operators \cite{Ka54} tells us that for $|J| \to \infty$,
	    \begin{equation}\label{Sh.5}
	  \tilde{E}^{(\infty)}(J;\xi) \sim
	  \exp \Big ( |J| \int_{-\infty}^\infty
	  \log (1 - \xi  \hat{K}^{(\infty)}(s)) \, ds
	  \Big ).
	  \end{equation}
	  
	  For the circulant correlation kernels (\ref{17.4c})
	  we have
	  \begin{equation}\label{Sh.6}
	  \hat{K}^{(\infty)}(s) =  {z\lambda^{(\infty)}(s)  \over 1 +  z  \lambda^{(\infty)}(s)}.
	  \end{equation} 
	  Substituting in  (\ref{Sh.5}) and comparing with
	  (\ref{17.4b}) shows, upon setting $\xi = 1$ 
	  and recalling (\ref{Sh.1}) that
	  \begin{equation}\label{Sh.6x}
	  {E}^{(\infty)}(0;J) \sim e^{- | J |  \beta P},
	   \end{equation}
	   in accordance with the functional form expected for the asymptotic gap probability of a general compressible fluid  \cite{FP92,Fo93}. 
	   
	   \subsection{Small separation form of $\rho^{(2,\infty)}(X,Y)$}
	   As a rewrite of (\ref{3.7}) we have
	    \begin{equation}\label{3.7a}
	    \rho_{(2)}^{(\infty)}(X,Y) = K^{(\infty)}(X,X) K^{(\infty)}(Y,Y) - | K^{(\infty)}(X,Y) |^2,
	    \end{equation}
	    telling us in particular that $\rho_{(2)}^{(\infty)}(X,Y)$ vanishes as $X \to Y$.
	    Substituting (\ref{17.4c}) in (\ref{3.7a}) and expanding to leading order in $z$ shows
	    \begin{equation}\label{3.7b}
	    \rho^{(2,\infty)}(X,Y) = 
	    z^2 \bigg ( \Big ( \int_{-\infty}^\infty \lambda^{(\infty)}(s) \, ds \Big )^2-
	    \Big |  \int_{-\infty}^\infty e^{2 \pi i s (X - Y)} \lambda^{(\infty)}(s) \, ds \Big |^2 \Big )
	    \bigg ) + O(z^3).
	    \end{equation}
	    
	    We see that if $\lambda^{(\infty)}(s)$ decays fast enough at infinity to allow the complex exponential to be expanded to second order, the small distance form of 
	    $\rho_{(2)}^{(\infty)}(X,Y)$ will always decay like a quadratic. In fact this will happen without first expanding to low order in $z$.

	   \begin{proposition}
	   	Let $\lambda^{(\infty)}(s)$ and $\lambda^{(\epsilon,\infty)}(s)$ decay at least as fast as order $1/|s|^3$ at infinity. Then
	   	$\rho_{(2)}^{(\infty)}(X,Y)$ goes to zero like $(X-Y)^2$
	   	as $X \to Y$.
	   \end{proposition}
	   
	   \begin{proof}
	   	The assumption on the decay of $\lambda^{(\infty)}(s)$ (and its counterpart in the complex Hermitian case) allows the complex exponential in
	   	(\ref{17.4c}) to be expanded to second order in $(X-Y)$.
	   	In the setting of Proposition \ref{p2.5b} the term proportional to $(X-Y)$ vanishes due to the parity of the integrand. Furthermore, substituting in the RHS of (\ref{3.7}) shows that the term independent of $(X-Y)$ also vanishes, leaving the term proportional to $(X-Y)^2$ as the leading term. In the setting of Proposition \ref{p2.e}, applying the same expansion in the analogue of (\ref{17.4c}) we see that the term proportional to $(X-Y)$ no longer vanishes, but nonetheless when substituted in (\ref{3.7}) its contribution to  the expansion of the RHS at this order cancels as does that of the term independent of $(X-Y)$, again leaving  the term proportional to $(X-Y)^2$ as the leading term.
	   \end{proof}
   
   To leading order in $z$, we see from (\ref{3.7b})
that the behaviour of $\rho^{(2,\infty)}(X,Y)$ will be determined by the behaviour of the Fourier transform of
$\lambda^{(\infty)}(s)$. Now from (\ref{17.4})
 \begin{equation}\label{3.7c} 	
 \int_{-\infty}^\infty \lambda^{(\infty)}(s) e^{-2 \pi i s Y} \, ds = g(Y),
 \end{equation}
 so we see that the small distance functional form in
 (\ref{3.7b}) is determined by the rate of vanishing of
 $g(0) - g(Y)$ as $Y \to 0$, which for 
 $\lambda^{(\infty)}(s)$ decaying slower that
 $O(1/s^3)$ will be slower than of order $Y^2$.

\subsection{Example of $g(u)$ a Gaussian}\label{gG}

Suppose we take for the functional form (\ref{15.6}) of the matrix elements
 \begin{equation}\label{gu}
 g(u) = {1 \over \sqrt{c}} e^{-\pi u^2/c}, \quad c>0.
 \end{equation}
Substituting in 
   (\ref{17.4}) shows 
 \begin{equation}\label{gu4}
 \lambda^{(\infty)}(s) = e^{-\pi c s^2}.
  \end{equation}
  This Gaussian functional form in $s$, decaying faster than any power and being analytic, implies upon repeated integration by parts in (\ref{17.4c}) that
   $K^{(\infty)}(X,Y)$ decays faster than any inverse power for $|X-Y|$ large, as will be the case whenever $\lambda^{(\infty)}$ has these properties. Specifically for (\ref{gu4}), this fast decay can be exhibited by first power series expanding in $z$
  and then computing the integrals by completing the square to obtain
  \begin{equation}\label{gu5} 
 K^{(\infty)}(X,Y) = - 
 \sum_{p=1}^\infty {(-z)^{p} \over \sqrt{c p}}
 e^{- \pi (Y - X)^2/cp}, \quad |z| < 1.
  \end{equation}
  Thus in fact there is term-by-term Gaussian decay. 
  Parametrise $c$ and $z$ in terms of $\beta$ and $\mu$
  according to
 \begin{equation}\label{tcp}
 c = 4 \pi \beta, \qquad z = e^{\beta \mu}.
  \end{equation}
  After a simple change of variables in
  		(\ref{17.4c}) the expression (\ref{K.14}) noted in the Introduction
		results, which is the
  		correlation kernel for free
  			fermions (or equivalently hard-core bosons) in one-dimension at inverse temperature $\beta$
  			and chemical potential $\mu$ \cite{Le66}. 
			
Some insight into the relation to free fermions at finite temperature can be obtained
by considering $L^{(L)}(X,Y)$ as specified in terms of $g$ by (\ref{15.6a+}) in the
limit $L \to \infty$ when it reads $L^{(\infty)}(X,Y) = g(X - Y)$. From (\ref{C.0a})
we then have
\begin{equation}\label{U.0}
p(\mathbf X_N) \propto \det [ L^{(\infty)}(X_j, X_k) ]_{j,k=1}^N \propto
\det [ e^{- \pi (X_j - X_k)^2/c} ]_{j,k=1}^N.
\end{equation}	
Now replace $\pi/c$ by $1/(2(1 - q^2))$ and observe that to leading order
\begin{equation}\label{U.2}
\det \Big [ 	e^{ - {1 \over 2} {1 \over 1 - q^2} (x_i - x_j)^2} \Big ]_{i,j=1}^N 	\mathop{\sim}\limits_{q \to 1} 
\det \Big [ 	e^{ - {1 \over 4} {1 + q^2 \over 1 - q^2} (x_i^2 + x_j^2) + {q \over 1 - q^2} x_i x_j}
\Big ]_{i,j=1}^N 	\mathop{\sim}\limits_{q \to 1} 
\end{equation}	
as follows by completing the square on the RHS.
The significance of this is that it is well known that the
RHS  of (\ref{U.2})	 can be rewritten in the form of the probability
density function for $N$ free fermions in a harmonic well on a line, in equilibrium
at a finite temperature; see \cite[\S 1]{LW20} for a clear derivation.
After appropriate scaling, 	the bulk correlation kernel is precisely
(\ref{K.14}) \cite{MNS94,GV03,Jo07,DDMS16}.

%At this stage, let us recall the matrix model of Moshe, Neuberger and Shapiro
%\cite{MNS94}, which is a probability density function on the space of
%Hermitian matrices $\{M\}$ specified by
%\begin{equation}\label{U.1}
%\int e^{ - {1 \over 2} {1 - q \over 1 + q} {\rm Tr} \, M^2}
%e^{ - {1 \over 2} {q \over 1 - q^2} {\rm Tr}([U,M][U,M]^\dagger)} \, d \mu(U),
%\end{equation}
%where $d \mu(U)$ is the Haar measure on space of unitary matrices.
%Using the well known Harish-Chandra/ Itzykson-Zuber matrix integral
%(see e.g.~\cite[Prop.~11.6.1]{Fo10}) the integral in (\ref{U.1}) can
%be computed explicitly, implying that the corresponding eigenvalue
%probability density function is proportional to
%\begin{equation}\label{U.2}
%\det \Big [ 	e^{ - {1 \over 4} {1 + q^2 \over 1 - q^2} (x_i^2 + x_j^2) + {q \over 1 - q^2} x_i x_j}
%\Big ]_{i,j=1}^N 	\mathop{\sim}\limits_{q \to 1} 
%\det \Big [ 	e^{ - {1 \over 2} {1 \over 1 - q^2} (x_i - x_j)^2} \Big ]_{i,j=1}^N;
%\end{equation}			
%cf.~(\ref{U.0}). The significance of the coincidence of the functional forms
%(\ref{U.0}) and the right hand side of (\ref{U.2}) is that it is well known that the
%left hand side of (\ref{U.2})	 can be rewritten in the form of the probability
%density function for $N$ free fermions in a harmonic well on a line, in equilibrium
%at a finite temperature; see \cite[\S 1]{LW20} for a clear derivation.
%After appropriate scaling, 	the bulk correlation kernel is precisely
%(\ref{K.14}) \cite{GV03,Jo07,DDMS16}.
			
		It is instructive to consider the $\beta \to \infty$ limit of (\ref{K.14}). We see
  			\begin{equation}\label{K.14a}
  			K^{(\infty)}(X,Y) \Big |_{\beta \to \infty} =
  			{1 \over 2 \pi} \int_{-\sqrt{\mu}}^{\sqrt{\mu}} e^{i (Y - X) k} \, dk =
  				{\sin \sqrt{\mu} (Y - X) \over \pi (Y - X)}.
  				\end{equation}
  				Setting $X = Y$ gives for the particle density the
  				value $\rho_{(1)}^{(\infty)} = \sqrt{\mu}/\pi$. Taking this to equal unity as a
  				normalisation, we then recognise (\ref{K.14a}) as the sine kernel from
  				random matrix theory as applies to the bulk of unitary invariant ensembles
				(see e.g.~\cite[Ch.~7]{Fo10}), the latter having a well known analogy
				with the ground state of free fermions on a line (see e.g.~\cite{DDMS19}).

 In relation to the gap probabilities associated with (\ref{K.14}), from (\ref{Sh.3}) these are determined by the eigenvalues of $\mathbb K_J$. With $J = (-x,x)$, in the present setting these satisfy
 \begin{equation}\label{J.1}
 \int_{-x}^x K^{(\infty)}(X-Y,0) f_j(Y) \, dY = \lambda_j^{(J)} f_j(X),
 \end{equation}
 where $\{ f_j(X)\}$ are the corresponding eigenfunctions.
 Taking the Fourier transform of both sides with respect to $X$, and also writing $f_j(X)$ in terms of its Fourier transform, this can be rewritten
  \begin{equation}\label{J.2}
  \hat{K}^{(\infty)}(k) \int_{-\infty}^\infty
  { \sin (x (k - s)) \over \pi (k - s)}
  \hat{f}_j(s) \, ds = \lambda_j^{(J)} \hat{f}_j(k).
  \end{equation}
  Consider this equation
  with 
  \begin{equation}\label{J.2a}
  \hat{f}_j(k) \mapsto ( \hat{K}^{(\infty)}(k))^{1/2}  \hat{f}_j(k),  
  \end{equation}
  and read off the explicit form of $\hat{K}^{(\infty)}(k)$ from (\ref{K.14}).
   This tells us that we have the identity \cite{IIK90}
   \begin{equation}\label{J.3}
   \prod_{j=0}^\infty (1 - \xi \lambda_j^{(J)})
  = \prod_{j=0}^\infty (1 - \xi \tilde{\lambda}_j^{(J)}),
  \end{equation}
  where $\{ \tilde{\lambda}_j^{(J)} \}$ are the eigenvalues of the integral operator on all of $\mathbb R$ with kernel
  \begin{equation}\label{J.4}
 \tilde{K}^{(\infty)}(k,s) 
= \Big ( {1 \over e^{\beta (k^2 - \mu)} + 1} \Big )^{1/2}
 { \sin (x (k - s)) \over \pi (k - s)}
 \Big ( {1 \over e^{\beta (s^2 - \mu)} + 1} \Big )^{1/2}.
 \end{equation}
 
 The significance of the functional form (\ref{J.4}), in contrast to  (\ref{K.14}), is that the former has the structure of a so-called integrable kernel \cite{IIKS90}.
 Associated with integral kernels are differential equations.
 Explicitly, with (\ref{J.3}) denoted by
 $\Delta(x,\mu,\xi)$, we have that
 $\sigma(x,t,\xi) := \log \Delta(\sqrt{\beta}x,t/\beta,\xi)$ satisfies the partial differential equation  \cite{IIK90}
  \begin{equation}\label{J.5} 
  (\partial_t \partial_x^2 \sigma)^2 = - 4
  (\partial_x^2 \sigma) \Big (
  2 x \partial_t \partial_x \sigma +
  (\partial_t \partial_x \sigma)^2 - 2 \partial_t \sigma
  \Big ),
  \end{equation}
  subject to the small-$x$ expansion
 \begin{equation}\label{J.6}
 \sigma(x,t,\xi) = - {\xi  \over \pi} \Big ( \int_{-\infty}^\infty {d \lambda \over 1 +
 e^{\lambda^2 - t} } \Big )x -
{\xi^2 \over 2 \pi^2}  \Big ( \int_{-\infty}^\infty {d \lambda \over 1 +
	e^{\lambda^2 - t} } \Big )^2 x^2 + \cdots
\end{equation}  

Consider now the gap probability in the $\beta \to \infty$ limit.
We know from (\ref{K.14a}) that 
\begin{equation}
\lim_{\beta \to \infty} (\sqrt{t}/\beta)  \tilde{K}^{(\infty)}(\sqrt{t}k/\beta,\sqrt{t} s/\beta) \Big |_{x \mapsto \sqrt{\beta}x \atop
\mu =  t/\beta}  = {\sin ( \sqrt{t}x (k - s)) \over \pi (k - s)},
\end{equation}
supported on $k,s \in [-1, 1]$. Thus in this limit $\sigma(x,t,\xi)$ depends on $x,t$ only through quantity $\tau := \sqrt{t} x$ and
moreover
\begin{equation}\label{J.5a}
\sigma(x,t,\xi) \to \log \det (\mathbb I - \xi \mathbb K^{(\tau,\rm sine)}_{(-1,1)}) =: \log \Delta(\tau,\xi),
\end{equation}
where $\mathbb K^{(\tau,\rm sine)}_{(-1,1)})$ denotes the integral operator supported on $(-1,1)$ with
kernel (\ref{K.14a}) and $\sqrt{\mu} = \tau$.  As noted in \cite{IIKS90}, with
\begin{equation}
\sigma_0(\tau,\xi) := \tau \partial_\tau \log  \Delta(\tau,\xi)
\end{equation}
it follows from (\ref{J.5}), with prime denoting differentiation with respect to $\tau$, that
\begin{equation}\label{J.5b}
(\tau \sigma_0'')^2 = - 4 (\tau \sigma_0' - \sigma_0) (4 \tau \sigma_0' + (\sigma_0')^2 - 4 \sigma_0).
\end{equation}
This nonlinear second order differential equation, which relates to the Hamiltonian theory of the Painlev\'e V system, was 
obtained for the Fredholm determinant for the sine kernel in (\ref{J.5a}) by the Kyoto school
\cite{JMMS80}; see also \cite[\S 8.3.5]{Fo10}.

\subsection{Example of $h(u) = 1/u$} 
Before specialising Proposition \ref{p2.e}, it is of interest to make note of the evaluation of the determinant in (\ref{C.0}) which holds for this choice of $h$.
This requires use of the Cauchy double alternant identity
(see e.g.~\cite{Ai56})
\begin{equation}\label{CDA}
\det \Big [ {1 \over x_j - y_k} \Big ]_{j,k=1}^N =
{\prod_{1\le j < k \le N}(x_j - x_k) (y_k - y_j) \over
\prod_{j,k=1}^N (x_j - y_k) 
}.
\end{equation}
Noting that with 
\begin{equation}
x_j = {1 \over 2i} e^{2 \pi i X_j/L}
e^{-2 \pi \epsilon/L}, \quad
y_j = {1 \over 2i} e^{2 \pi i X_j/L}
e^{2 \pi \epsilon/L},
\end{equation}
we have
\begin{equation}\label{CDA.1}
\det \Big [
{\pi \over L \sin ( \pi (X_j - X_k + 2 i \epsilon)/L)}
\Big ]_{j,k=1}^N = \Big ( {\pi \over L} \Big )^N \prod_{j=1}^N e^{-2 \pi i N X_j/L}
\det \Big [{1 \over x_j - y_k} \Big ]_{j,k=1}^N,
\end{equation} 
application of (\ref{CDA.1}) gives for $h(u) = i/u$,
\begin{equation}\label{CDA.2}
\det [ L(X_j, X_k)]_{j,k=1}^N =
\Big ({\pi i \over L} \Big )^N  { \prod_{1 \le j < k \le N}
	(\sin (\pi (X_k - X_j)/L))^2 \over
 \prod_{j,k=1}^N \sin(\pi (X_j - X_k + 2 i \epsilon)/L)}.
\end{equation}	
First considered in \cite{Ga66}
(see also \cite{Fo93a}), (\ref{CDA.2}) corresponds to the Boltzmann factor for a statistical mechanical system of $N$ particles (two-dimensional charges interacting via a logarithmic potential) in equilibrium at inverse temperature $\beta = 2$, confined to the interval $[0,L]$ on the $x$-axis in periodic boundary conditions and in the presence of a perfect conductor at $y = \epsilon$. For each charge at $(X,0)$, the  perfect conductor creates an image charge of opposite sign at $(X,2\epsilon)$. We remark that in the reference \cite{Fo93a}, the PDF
corresponding to (\ref{CDA.2}) is related to the  theory of parametric eigenvalue motion due to Pechukas \cite{Pe83} and Yukawa \cite{Yu86}, as discussed
extensively in the book on quantum chaos by Haake \cite{Ha92}.

Returning now to Proposition \ref{p2.e}, setting $h(u) = 1/u$ in (\ref{e.1}) gives
\begin{equation}\label{e.2}
\lambda^{(\epsilon, \infty)}(s):= i \int_{-\infty}^{\infty}
h  (  - Y + 2i \epsilon  ) e^{2 \pi i s Y} \, d Y
= 2 \pi \left \{ \begin{array}{ll} e^{-4 \pi \epsilon s}, \qquad s \ge 0 \\
0, \qquad {\rm otherwise}. \end{array} \right. 
\end{equation}
Substituting this for $\lambda^{(\infty)}(s)$ in
  (\ref{17.4b}) and (\ref{17.4c}) specifies the corresponding pressure and correlation kernel.
Specifically, for the latter \cite{Ga66}
\begin{equation}\label{e.3}
K^{(\epsilon, \infty)}(X,Y) = 
\int_0^\infty {e^{2 \pi i (X - Y) s}
\over (1/2 \pi z) e^{4 \pi  \epsilon s} + 1} \, ds.
\end{equation}	  

From the viewpoint of universal forms for two-point correlations in Coulomb systems of
restricted dimension \cite{FJT96}, of relevance is the large $X$ asymptotic form of (\ref{e.3}),
which is uniform for large $\epsilon$, when furthermore $z$ is related to $\epsilon$ by
\begin{equation}\label{z.1}
(1/2 \pi z) = e^{-4 \epsilon h}, \quad h > 0.
\end{equation}
Integrating by parts once, then extending the domain of integration to all of $\mathbb R$ in the
resulting integral gives
\begin{equation}\label{z.2}
K^{(\epsilon, \infty)}(X,Y)  \mathop{\sim}\limits_{|X - Y| \to \infty} {1 \over 2 \pi i} \bigg ( - {1 \over X - Y} - {e^{2 i h (X - Y)} \over
(2 \epsilon/\pi) \sinh(\pi(X - Y)/2 \epsilon) } \bigg ).
\end{equation}
This substituted in (\ref{3.7}) implies
\begin{equation}\label{z.3}
\rho^{(2,\infty)}(X,Y) -
(\rho^{(1,\infty)})^2  \mathop{\sim}\limits_{|X - Y| \to \infty}^\cdot - {1 \over 4 \pi^2 (X - Y)^2} - {1 \over 8 \epsilon^2 \sinh^2 \pi (X - Y)/2 \epsilon}.
\end{equation}
Here the modification of the asymptotic symbol $  \mathop{\sim}\limits^\cdot$ indicates that oscillatory terms averaging to zero are ignored.
This is the universal form predicted in \cite[Eq.~(3.4)]{FJT96}.

Write now relate $z$ to $\epsilon$ by (\ref{z.1}) and take the limit $\epsilon \to \infty$ with $X,Y$ fixed. We see that
\begin{equation}\label{e.4}
K^{(\epsilon, \infty)}(X,Y) 
 \Big |_{(1/2 \pi z) = e^{-2 \epsilon h/ \pi} \atop \epsilon \to \infty}
= 
e^{i (X - Y) h} {\sin (h (X - Y)) \over \pi (X - Y)}.
\end{equation}
Substituting in the determinant formula (\ref{Sh.1}) we see that the contributions from the factors of the form $e^{i (X - Y) h}$ cancel and as in (\ref{K.14a}) the sine kernel from random matrix theory is reclaimed, as already 
known from \cite{Ga66}.

\begin{remark}
Although not considered further in the present work, we note that the choice $h(u) = 1 / \sinh u$ is also
of interest from the viewpoint of the study \cite{FJT96}. Moreover, as observed in that latter reference,
in the limit $L \to \infty$, an identity analogous to (\ref{CDA.2}) holds true, showing that their is an underlying
pair potential.
\end{remark}

\subsection{Higher dimensions}
The appropriate generalisation of the circulant $L$-ensemble structure to higher dimensions ---
say to a particle system confined to a cube $[0,L]^d$ with periodic boundary conditions --- is to define
vectors $\mathbf x = (x_1,\dots, x_d)$ (and similarly $\mathbf y$) and extend the definition
(\ref{15.6a+})  to read
\begin{equation}\label{kL}
L^{(L)}(\mathbf x, \mathbf y) = g \Big ( (L/\pi) \sin( \pi (x_1 - y_1)/L), \dots,
(L/\pi) \sin( \pi (x_d - y_d)/L) \Big ).
\end{equation}
With this done, the probability density function for a configuration $\mathbf X_N$ in the
cube is again given by (\ref{C.0a}), but with $\mathbb L$ in (\ref{C.0b}) 
now specified on the interval $[-L/2,L/2]^d$ 
with kernel
(\ref{kL}),
\begin{equation}\label{kLa}
\mathbb L[f](\mathbf{x}) = \int_{[-L/2,L/2]^d}
g\Big ((L/\pi) \sin( \pi (x_1 - y_1)/L), \dots,
(L/\pi) \sin( \pi (x_d - y_d)/L) \Big ) f(\mathbf y) \, d\mathbf y.
\end{equation}
The normalised eigenfunctions are
\begin{equation}\label{kLb}
v_{\mathbf p}(\mathbf x) = {1 \over L^{d/2}}  \prod_{j=1}^d e^{2 \pi i p_j x_j /L}
= {1 \over L^{d/2}} e^{2 \pi i  \mathbf p \cdot \mathbf x}
, \quad p_j \in \mathbb Z \: (j=1,\dots,d).
\end{equation}

Using the multidimensional analogues of (\ref{2.39}) and (\ref{2.40}), the results of Proposition
\ref{p2.5b} can be extended to higher dimensions.

\begin{proposition}\label{p2.5c}
Consider the $d$-dimensional determinantal point-process of the $L$-ensemble type
specified by (\ref{kL}) and surrounding text. Let
\begin{equation}\label{kLc}
    \lambda^{(d,\infty)}(\mathbf s)  = \int_{\mathbb R^d} {g}(\mathbf t) e^{2 \pi i \mathbf s \cdot \mathbf t} \, d\textbf t.
 \end{equation}      
 We have 
  \begin{equation}\label{kLd} 
  \beta P =
\int_{\mathbb R^d} \log ( 1 + z \lambda^{(d,\infty)}(\mathbf s)   ) \, d \mathbf s,
  \end{equation}
  and
  \begin{equation}\label{kLe}     
  K^{(d,\infty)}( \mathbf x,\mathbf y)   =   
  {z } \int_{\mathbb R^d}  {e^{2 \pi i (\mathbf y - \mathbf x) \cdot \mathbf s} \lambda^{(d,\infty)}(\mathbf s)  \over 1 +  z  \lambda^{(d,\infty)}(\mathbf s)} \, d\mathbf s.
 \end{equation} 
 \end{proposition} 
 
 As an explicit example, consider the $d$-dimensional generalisation of Gaussian (\ref{gu}),
 \begin{equation}
 g(\mathbf u) = {1 \over c^{d/2}} e^{- \pi  \mathbf u^2/c}, \quad c>0.
 \end{equation} 
Substituting in (\ref{kLc}) gives
 \begin{equation}
    \lambda^{(d,\infty)}(\mathbf s)  =  e^{- \pi c \mathbf s^2}.
    \end{equation}
The formula  (\ref{kLd}) for the pressure, from the spherical symmetry of the integrand,
then simplifies upon the use of  polar coordinates to read
   \begin{equation}\label{kLf}  
   \beta P = |\Omega_d| \int_0^\infty   r^{d-1} \log  ( 1 + z   e^{- \pi c r^2} ) \, d r,
 \end{equation}
 where $ |\Omega_d| $ denotes the surface area of the unit ball in $d$-dimensions,
while the formula (\ref{kLe}) for the correlation kernel reads
  \begin{equation}\label{kLg}
  K^{(d,\infty)}( \mathbf x,\mathbf y)   =  \Big ( {1 \over 2 \pi } \Big )^d \int_{\mathbb R^d} {e^{i(\mathbf y - \mathbf x) \cdot \mathbf k} \over
  e^{\beta( \mathbf k^2 - \mu)} + 1} \, d \mathbf k.
   \end{equation}  
   In (\ref{kLg}) the parameters $\beta$ and $\mu$ have been introduced in favour of $c$ and $z$ as in (\ref{K.14}).
   The resulting expression can be recognised as the correlation kernel for free fermions in $d$-dimensions in equilibrium
   at inverse temperature $\beta$ and chemical potential $\mu$ (see \cite{DDMS16}, where it is furthermore noted
   that the introduction of polar coordinates can be used to reduce (\ref{kLg}) down to a one-dimensional integral involving
   a Bessel function).
   
   A question of interest is the asymptotic form of the probability that there are no particles in a region $\Lambda$ say of
   $\mathbb R^d$ --- what was termed in the one-dimensional case in Section \S \ref{S3.2} as the gap probability,
   but what in higher dimensions is usually referred to as the hole probability. The fact that 
   $K^{(d,\infty)}( \mathbf x,\mathbf y) $ only depends on the differences of the components allows for the determination of the limiting asymptotic form for
   $|\Lambda| \to \infty$ \cite{Wi60}
   \begin{equation}\label{Laf}
   E(0,\Lambda) \sim e^{-|\Lambda| \beta P},
   \end{equation}
   where $\beta P$ is given by (\ref{kLf}), in keeping with 
   (\ref{Sh.6x}). Generally this asymptotic behaviour of the hole probability is expected whenever the particle system is compressible \cite{FP92}. Note that the latter condition
   ceases to hold in the zero temperature, $\beta \to \infty$, limit of (\ref{kLg}). For results on the corresponding asymptotic form of the hole probability, see the recent work \cite{GDS21}.
   
   Taken literally the complex Hermitian circulant matrix construction (\ref{C.0}) does not have a
   generalisation to higher dimension due to the use of the complex unit $i$ as effectively extending
   from the real line to the $xy$-plane. However, by extending the $L$-ensemble formalism from
   one to two-components,  a two-dimensional complex Hermitian Toeplitz construction
    is known from the work of Gaudin on the
   two-dimensional two-component Coulomb gas at a special coupling \cite{Ga85, Fo98a}.
   
   \subsection*{Afterword}
   My earliest memory of encountering the work of F.~Haake has through his work with Grobe and Sommers on the hole
   probability in the Ginibre ensemble \cite{GHS88}. I was able to use this to deduce the first four terms in
   its asymptotic expansion \cite{Fo92b}. Around the same time the first edition of F.~Haake's celebrated book
   {\it Quantum signatures of chaos} \cite{Ha92} appeared. In addition to be taken by the discussion relating the Ginibre
   ensemble to dissipative quantum systems, I payed particular attention to the sections 
   on the so-called Pechukas--Yukawa gas, and Dyson's Brownian-motion model, in the chapter on Level Dynamics.
   The latter was very influential in shaping my own subsequent work on the topic \cite[Ch.~11]{Fo10}.

\subsection*{Acknowledgements}
	This research is part of the program of study supported
	by the Australian Research Council Centre of Excellence ACEMS, and the project DP210102887.
	The presentation has benefitted from a number of considered remarks put forward by the referees.

\providecommand{\bysame}{\leavevmode\hbox to3em{\hrulefill}\thinspace}
\providecommand{\MR}{\relax\ifhmode\unskip\space\fi MR }
\providecommand{\MRhref}[2]{%
  \href{http://www.ams.org/mathscinet-getitem?mr=#1}{#2}
}
\providecommand{\href}[2]{#2}


\begin{thebibliography}{10}


 \bibitem{Hu40} K.~Husimi, \emph{ Some formal properties of the density matrix}, Proc. Phys. Math. Soc. Jpn. \textbf{22} (1940), 264--314 .
 
  \bibitem{Ma75}
 O. Macchi.  \emph{ The coincidence approach to stochastic point processes}, Adv. Appl.
Probab.,  \textbf{7} (1975), 83--122.

\bibitem{Dy62a}
F.J.~Dyson, \emph{Statistical theory of energy levels of complex systems {III}},
  J. Math. Phys. \textbf{3} (1962), 166--175.

  \bibitem{Gi65}
J.~Ginibre, \emph{Statistical ensembles of complex, quaternion, and real
  matrices}, J. Math. Phys. \textbf{6} (1965), 440--449.
  
  \bibitem{AJ81}
A.~Alastuey and B.~Jancovici, \emph{On the two-dimensional one-component
  {Coulomb} plasma}, J. Physique \textbf{42} (1981), 1--12.
  
   \bibitem{Ga66}   
  M. Gaudin, \emph{Une famille \`a un param\`etre d'ensembles unitaires}, Nucl.
Phys. \textbf{85} (1966), 545--575.
  
  
  \bibitem{Ga85} 
  M. Gaudin,  \emph{L'isotherme  critique  d'un  plasma sur  r\'eseau ($\beta= 2$, $d= 2$, $n= 2$)}, J. Physique \textbf{46} (1985),  1027--1042.
  
 \bibitem{Ka67} P.W.~Kasteleyn, \emph{Graph theory and crystal physics}, 
  In "Graph Theory and Theoretical Physics", pgs. 43--110, Academic Press, London, 1967.


  \bibitem{Dy62b}
F.J. Dyson, \emph{A {B}rownian motion model for the eigenvalues of a random
  matrix}, J. Math. Phys. \textbf{3} (1962), 1191--1198.
  
    \bibitem{Fi84}
M.E. Fisher, \emph{Walks, walls, wetting, and melting}, J. Stat. Phys.
  \textbf{34} (1984), 667--729.
  
    
   \bibitem{LM66}
E. ~Lieb  and  D.C. ~Mattis,  \emph{Mathematical  physics  in  one  dimension},  Academic  Press,  NewYork, 1966.


 \bibitem{Me91}
M.L. Mehta, \emph{Random matrices}, 2nd ed., Academic Press, New York, 1991.


  \bibitem{Fo98a}
P.J. Forrester, \emph{Exact results for two-dimensional {Coulomb} systems},
  Phys. Reports \textbf{301} (1998), 235--270.

  \bibitem{BO00}  
  A. Borodin and G. Olshanski, 
  \emph{Distributions on partitions, point processes and the hypergeometric kernel}, Comm. Math. Phys.
  \textbf{211} (2000), 335--358.
  
   \bibitem{Jo02}
K.~Johansson, \emph{Non-intersecting paths, random tilings and random matrices},
  Prob. Theory Related Fields \textbf{123} (2002), 225--280.

\bibitem{Bo06}  
  N.M.~Bogoliubov, \emph{XX Heisenberg chain and random walks}, J. Math. Sci. \textbf{138} (2006), 5636--5643.

  
 \bibitem{HKPV08}
J.B. Hough, M.~Krishnapur, Y.~Peres, and B.~Vir\'ag, \emph{Zeros of {G}aussian
  analytic functions and determinantal point processes}, American Mathematical
  Society, Providence, RI, 2009.
  

\bibitem{TSZ08} 
 S.~Torquato and A.~Scardicchio and C.E.~Zachary,
 \emph{Point processes in arbitrary dimension from fermionic gases, random matrix
 theory, and number theory}, J.~Stat.~Mech. (2008), P110019.

 \bibitem{Fo10}
P.J. Forrester, \emph{Log-gases and random matrices}, Princeton University
  Press, Princeton, NJ, 2010. 
  
  \bibitem{Ka16}
M.~Katori, \emph{Bessel processes, {S}chramm--{L}oewner evolution, and the {D}yson
  model}, Springer briefs in mathematical physics, vol.~11, Springer, Berlin,
  2016.
  
   \bibitem{DDMS19} D.S.~Dean, P.~Le Doussal,
S.N.~Majumdar and G.~Schehr, \emph{Non-interacting fermions in a trap and random matrix theory}, J.~Phys.~A \textbf{52},
(2019), 144006.

\bibitem{So00}
A.~Soshnikov, \emph{Determinantal random point fields}, Russian Math. Surveys
  \textbf{55} (2000), 923--975.
  
  \bibitem{DV03}
  D. J. Daley and D. Vere-Jones, \emph{An introduction to the theory of point processes} Vol. I. 
   Springer-Verlag, New York, second edition, 2003. Elementary
theory and methods.
  
  \bibitem{Ly03}
R. Lyons, \emph{ Determinantal probability measures}, Publ. Math. Inst. Hautes Etudes Sci., 
\textbf{98} (2003), 167--212.

   \bibitem{ST03} 
  T. Shirai and Y. Takahashi, \emph{Random point fields associated with certain Fredholm
determinants. I. Fermion, Poisson and boson point processes}, J. Funct. Anal., 205 (2003)
414--463.

  \bibitem{BDF10}
  A. Borodin, P. Diaconis, and J. Fulman, \emph{On adding a list of numbers (and other one-dependent determinantal processes)}, 
  Bull. Am. Math. Soc. \textbf{47} (2010), 639--670.

 \bibitem{Bo11}
A.~Borodin, \emph{Determinantal point processes}, The {O}xford {H}andbook of {R}andom {M}atrix
  {T}heory (G.~Akemann, J.~Baik, and P.~di~Francesco, eds.), Oxford University
  Press, Oxford, 2011, pp.~231--249.
  
 \bibitem{KT12}
 A.~Kulesza and B.~Taskar,  \emph{Determinantal point processes for machine learning},
 Found. Trends Mach. Learn. \textbf{5} (2012), 123--286.  
 
  \bibitem{LMR15}
 F. Lavancier, J. M\o{}ller, and E. Rubak, \emph{Determinantal point process models and statistical inference},
 J.  Royal  Stat. Soc.:  Series  B  \textbf{77} (2015), 853--877.

\bibitem{PS11}
L.~Pastur and M.~Shcherbina, \emph{Eigenvalue distribution of large random
  matrices}, American Mathematical Society, Providence, RI,, 2011.
  
  \bibitem{BQ17} 
  A.I. Bufetov and Y. Qiu, 
  \emph{Determinantal point processes associated with Hilbert spaces of holomorphic functions},
  Commun. Math. Phys. \textbf{351} (2017), 1--44.   
  
  \bibitem{KS19}
   M.~Katori and T.~Shirai,  \emph{Partial isometries, duality, and determinantal point processes},
   arXiv:1903.04945.
   
     \bibitem{Fo93a}
P.J.~Forrester, Statistical properties of the eigenvalue motion
of Hermitian matrices, Phys.~Lett.~A {\bf 173} (1993), 355--359.

 \bibitem{Br05} 
E. I.  Broman,  \emph{One-dependent trigonometric determinantal processes are two-block-factors}, Ann. Probab. \textbf{33} (2005), 601--609. 

 \bibitem{DGK09} 
   H. Dai, Z. Geary and L.P.~Kadanoff
\emph{Asymptotics of eigenvalues and eigenvectors of Toeplitz matrices},
J. Stat. Mech. , \textbf{2009} (2009), P05012.


\bibitem{MSS17}
 T. Mondal, S. Sadhukhan, and P. Shukla, \emph{Extended states with Poisson spectral statistics},
 Phys. Rev. E \textbf{95}, (2017), 062102. 
 
   \bibitem{Bo18}
  A.~Bose,  \emph{Patterned Random matrices}. Chapman and Hall/CRC, Boca Raton, FL, 2018.
  
  \bibitem{MS19}
  T. Mondal and P.Shukla, \emph{Statistical analysis of chiral structured ensembles: Role of matrix
constraints},
  Phys. Rev. E \textbf{99} (2019), 022124.
  
 \bibitem{Bo20}
  E. Bogomolny, \emph{Spectral  statistics  of  random  Toeplitz  matrices}, Phys. Rev. E \textbf{102} (2020), 04101(R).
  
  \bibitem{HM06}
 J.-P. Hansen and I.R. McDonald,
 \emph{Theory of Simple Liquids}, Academic Press, 3rd Edition, 2006.


  \bibitem{WW65}
E.T. Whittaker and G.N. Watson, \emph{A course of modern analysis}, 4th ed.,
  Cambridge University Press, Cambridge, 1927.
  
  
  \bibitem{Ka54}
   M.~Kac, \emph{Toeplitz matrices, translation kernels, and a related problem in probability theory}, Duke Math. J. \textbf{21} (1954), 501--509.   
   
  \bibitem{FP92}
P.J. Forrester and C.~Pisani, \emph{The hole probability in log-gas and random
  matrix systems}, Nucl. Phys. B \textbf{374} (1992), 720--740.
  
     \bibitem{Fo93}
  P.J.~Forrester, \emph{Log-gases, random matrices and the Fisher-Hartwig
conjecture}, J.~Phys.~A {\bf 26} (1993), 1179--1192. 


  \bibitem{Le66}
 A. Lenard, \emph{One-dimensional impenetrable bosons in thermal equilibrium}, J. Math. Phys. \textbf{7} (1966), 1268--1272.
 
 \bibitem{LW20}
K.~Liechty and D.~Wang, \emph{Asymptotics of free fermions in a quadratic well at finite temperature
and the Moshe--Neuberger--Shapiro random matrix model},
Ann. Inst. H. Poincar\'e  Probab. Statist. \textbf{56} (2020), 1072--1098. 

\bibitem{MNS94}
M.~Moshe, H.~Neuberger, and B.~Shapiro, \emph{Generalized ensemble of random
  matrices}, Phys. Rev. Lett. \textbf{73} (1994), 1497--1500.
  
  
  \bibitem{GV03}
A. M. Garcia-Garcia  and  J. J. M.  Verbaarschot,
\emph{Critical  statistics  in  quantum  chaos and Calogero-Sutherland model at finite temperature}, Phys. Rev. E
\textbf{67}, (2003) 046104. 

  
 \bibitem{Jo07}  
   K. Johansson,  \emph{From Gumbel to Tracy--Widom}, Prob. Theor. Rel. Fields \textbf{138} (2007), 75--112.

\bibitem{DDMS16} D.S.~Dean, P.~Le Doussal,
S.N.~Majumdar and G.~Schehr, \emph{Non-interacting fermions at finite temperature in a $d$-dimensional trap:  universal correlations}, Phys.~Rev.~A \textbf{94},
(2016), 063622.


\bibitem{DDMS19} D.S.~Dean, P.~Le Doussal,
S.N.~Majumdar and G.~Schehr, \emph{Non-interacting fermions in a trap and random matrix theory}, J.~Phys.~A \textbf{52},
(2019), 144006.

 \bibitem{IIK90} 
 A.R.   Its,, A.G. Izergin, and V. E. Korepin. \emph{Temperature correlators of the impenetrable Bose gas as an integrable system},
  Commun.  Math. Phys.  \textbf{129} (1990),  205--222.
 
 \bibitem{IIKS90}
A.R. Its, A.G. Izergin, V.E. Korepin, and N.A. Slavnov, \emph{Differential
  equations for quantum correlation functions}, Int. J. Mod. Phys B \textbf{4}
  (1990), 1003--1037.
  
  \bibitem{JMMS80}
M.~Jimbo, T.~Miwa, Y.~M\^ori, and M.~Sato, \emph{Density matrix of an
  impenetrable {Bose} gas and the fifth {Painlev\'e} transcendent}, Physica
  \textbf{1D} (1980), 80--158.


\bibitem{Ai56}
A.C. Aitken, \emph{Determinants and matrices}, 9th ed., Oliver and Boyd,
  Edinburgh \& London, 1956.
  
  
 \bibitem{Pe83}
  P.~Pechukas,  \emph{Distribution of energy eigenvalues in the irregular spectrum}, Phys. Rev. Lett., \textbf{51} (1983), 943--946.
  
  
\bibitem{Yu86}
  T. Yukawa, \emph{Lax form of the quantum eigenvalue problem},
  Phys. Lett. A \textbf{116} (1986) 227--230.
  
   \bibitem{Ha92}
F.~Haake, \emph{Quantum signatures of chaos}, Springer, Berlin, 1992.
  
  
  \bibitem{FJT96}
P.J. Forrester, B.~Jancovici, and G.~T\'ellez, \emph{Universality in some
  classical {Coulomb} systems of restricted domain}, J. Stat. Phys. \textbf{84}
  (1996), 359--378.

  
  \bibitem{Wi60}
  H.~Widom, \emph{A theorem on translation kernels in $n$ dimensions}, Trans. Amer. Math. Soc. \textbf{94} (1960), 170--180.  
 
  

  \bibitem{GDS21}
G.~Gouraud, P.~Le Doussal,
 and G.~Schehr, \emph{Hole probability for noninteracting fermions in a $d$-dimensional trap},
 arXiv:2104.08574.
 
 
   
 \bibitem{GHS88}
R. Grobe, F. Haake, and H.-J. Sommers, 
\emph{Quantum distinction of regular and chaotic dissipative motion},
Phys. Rev. Lett.61(1988) 1899.
  
  
  \bibitem{Fo92b}
  P.J.~Forrester, \emph{Some statistical properties of the eigenvalues of complex random matrices},
  Phys. Lett. A, \textbf{169}, 21--24.
  
  
 


  
  


  


  

    
  

  
  
  


  
 
  



 


 
 

 
 
 
  
 


 
 


 
  
   
  

 
  
  
    
  
  
 

 





 
 
 
 



 





  
 
 
   



 
  
 
  
  

  \end{thebibliography}
\end{document}